\documentclass[a4paper,USenglish]{lipics}

\usepackage{booktabs}
\usepackage{cite}

\newcommand{\SUBSTR}[2]{\ensuremath{T\lbrack #1.. #2\rbrack}}
\newcommand{\CARD}[1]{\ensuremath{\vert #1\vert}}
\newcommand{\SA}[1]{\ensuremath{\textsf{SA}\ifthenelse{\equal{#1}{}}{}{\lbrack#1\rbrack}}}

\newcommand{\UnaryOperator}[2][]{%
  \ifx&#1&%
  \ensuremath{\mathop{}\mathopen{}#2\mathopen{}}%
  \else%
  \ensuremath{\mathop{}\mathopen{}#2\mathopen{}\left(#1\right)}%
  \fi%
}
\newcommand{\Oh}[1]{\UnaryOperator[#1]{\mathcal{O}}}

\newcommand{\menge}[1]{\ensuremath{\left\{#1\right\}}} %
\newcommand{\gauss}[1]{\left\lfloor#1\right\rfloor} %
\newcommand{\abs}[1]{\left|#1\right|} %

\newcommand{\smallOh}[1]{\UnaryOperator[#1]{o}}

\title{On the Benefit of Merging Suffix Array Intervals for Parallel Pattern Matching}

\author[1]{Johannes Fischer}
\author[1]{Dominik K{\"o}ppl}
\author[1]{Florian Kurpicz} 
\affil[1]{%
Dept.\ of Computer Science, Technische Universit{\"a}t Dortmund, Germany\\
johannes.fischer@cs.tu-dortmund.de, dominik.koeppl@tu-dortmund.de, florian.kurpicz@tu-dortmund.de}

\authorrunning{J. Fischer, D. K{\"o}ppl, and F. Kurpicz} %

\Copyright{Johannes Fischer, Dominik K{\"o}ppl, and Florian Kurpicz} %

\serieslogo{}%
\volumeinfo%
  {}%
  {2}%
  {}%
  {1}%
  {1}%
  {1}%
\EventShortName{}
\DOI{10.4230/LIPIcs.xxx.yyy.p}%

\begin{document}

\maketitle 

\newcommand{\occ}{\text{occ}}

\begin{abstract}
We present parallel algorithms for exact and approximate pattern matching with suffix arrays, using a CREW-PRAM with $p$ processors.
Given a static text of length $n$, we first show how to compute the suffix array interval of a given pattern of length $m$ 
in $\Oh{\frac{m}{p}+ \lg p + \lg\lg p\cdot\lg\lg n}$ time for $p \le m$.
For approximate pattern matching with $k$ differences or mismatches, we show how to compute all occurrences of a given pattern 
in $\Oh{\frac{m^k\sigma^k}{p}\max\left(k,\lg\lg n\right)\!+\!(1+\frac{m}{p}) \lg p\cdot \lg\lg n + \occ}$ time, where $\sigma$ is the size of the alphabet and $p \le \sigma^k m^k$.
The workhorse of our algorithms is a data structure for merging suffix array intervals quickly:
Given the suffix array intervals for two patterns $P$ and $P'$, we present a data structure for computing the interval of $PP'$ in $\Oh{\lg\lg n}$ sequential time, or in $\Oh{1+\lg_p\lg n}$ parallel time.
All our data structures are of size $\Oh{n}$ bits (in addition to the suffix array).
\end{abstract}

\newcommand{\Head}  [1]{\UnaryOperator[#1]{\Gamma}}
\newcommand{\RLeaf} [1]{\UnaryOperator[#1]{\ensuremath{H_{\textup{R}}}}}
\newcommand{\LLeaf} [1]{\UnaryOperator[#1]{\ensuremath{H_{\textup{L}}}}}
\newcommand{\lcp}   [1][]{\UnaryOperator[#1]{\mathsf{lcp}}}
\newcommand{\I}   [1]{\ensuremath{\left\lbrack#1\right\rbrack}}
\newcommand{\Key}   [1]{\UnaryOperator[#1]{\mathsf{key}}}
\newcommand{\Value} [1]{\UnaryOperator[#1]{\mathsf{val}}}
\newcommand{\lab}   [1][]{\UnaryOperator[#1]{\mathsf{label}}}
\newcommand{\lca}   [1][]{\UnaryOperator[#1]{\mathsf{lca}}}
\newcommand{\timeSA}{\ensuremath{t_{\SA{}}}}
\newcommand{\spaceSA}{\ensuremath{\abs{\textup{CSA}}}}
\newcommand{\child}   [1][]{\UnaryOperator[#1]{\mathsf{child}}}
\newcommand{\pathlabel} [1][]{\UnaryOperator[#1]{\mathsf{pathlabel}}}
\newcommand{\leafrank} [1][]{\UnaryOperator[#1]{\mathsf{leafrank}}}
\newcommand{\Int}[1]{\UnaryOperator[#1]{\mathcal{I}}}
\newcommand{\IntJ}{\ensuremath{\mathcal{J}}}
\newcommand{\B}[1]{\ensuremath{\mathsf{b}(#1)}}
\newcommand{\E}[1]  {\ensuremath{\mathsf{e}(#1)}}

\section{Introduction}
We consider parallelizing indexed pattern matching queries in static texts, using (compressed) suffix arrays \cite{manber93suffix,navarro07compressed} and (compressed) suffix trees \cite{ohlebusch10cst++,sadakane07compressed} as underlying indexes.
We work with the \emph{concurrent read exclusive write} (CREW) \emph{parallel random access machine} (PRAM) with $p$ processors, as this model most accurately reflects the design of existing multi-core CPUs.
Our starting point is that a (possibly very long) pattern can be split up into several subpatterns that can be matched in parallel.
In a suffix array, this will result in $p$ intervals, each corresponding to one of the subpatterns.
These intervals, called \emph{subintervals}, will then be combined (using a merge tree approach) to finally yield the interval for the entire pattern.
From this interval, all occurrences of the pattern in the text could then be easily listed.

We also consider parallel indexed pattern matching with $k$ errors, again using the same indexes as in the exact case.
Here, we follow the approach of Huynh et al. \cite{Huynh2006}, whose basic idea is to first make all possible modifications of the pattern within distance $k$, and then match those modifications in the suffix array.
To avoid repeated computations of subintervals, a preprocessing is performed for every prefix and suffix of the pattern.
We show how to parallelize both steps (preprocessing and the actual matching), resulting in a fast parallel matching algorithm.
We stress that in the case of approximate pattern matching, parallel pattern matching algorithms are of even more practical importance than in the exact case, as this is an inherently time-consuming task in the sequential case, even for short patterns.

\subsection{Our Results}
In the abstract, we stated the results for uncompressed suffix arrays \cite{manber93suffix} as the underlying index, which requires $\Oh{n\lg n}$ bits of space for a text of length $n$.
However, there exists a wealth of compressed versions of suffix arrays (CSAs) \cite{navarro07compressed}, which are smaller (using $\spaceSA$ bits), but often have nonconstant access time $\timeSA$. (See also Figure~\ref{figSA} for known trade-offs.)
Here, we state our results more generally, using the parameters $\spaceSA$ and $\timeSA$.

Our first result (Thm.~\ref{the:parallel_exact}) is an index of size $\spaceSA + \Oh{n}$ bits that, with $p \le m$ processors, allows us to compute the suffix array interval of a pattern of length $m$ in $\Oh{\timeSA\left(\frac{m}{p} + \lg p + \lg\lg p\cdot\lg\lg n\right)}$ time 
and $\Oh{\timeSA\left( m+ \min\left(p,\lg n\right)\left(\lg p + \lg\lg p\cdot\lg\lg n\right)\right)}$ work.
Our second result (Thm.~\ref{the:k_error_approximate_string_matching}) is an index of the same size $\spaceSA + \Oh{n}$ bits that can find all $\occ$ occurrences of a pattern 
in $\Oh{\timeSA \left( \frac{m^k \sigma^k}{p} \max\left(k,\lg\lg n\right) + (1+\frac{m}{p}) \lg p \cdot \lg \lg n\right) + \occ}$ time, for $p \le m^k \sigma^k$.
Both results rely on the ability to merge two suffix array intervals quickly, a task for which we give a data structure of size \Oh{n} bits on top of CSA that allows us to do the merging in \Oh{\timeSA \lg \lg n} sequential (Lemma~\ref{lem:suffix_interval_merging_time}) or in \Oh{\timeSA (1+\lg_p \lg n)} parallel time (Lemma \ref{lem:parallel_merging}).

\subsection{Related Work}
We are only aware of one article addressing the parallelization of single queries \cite{Jekovec2015}. Their main result is to augment a suffix tree with a data structure of size $\Oh{n\lg p}$ words that answers pattern matching queries using $\Oh{\frac{m}{p}\lg p}$ time and $\Oh{m\lg p}$ work, which is worse than ours in all three dimensions. Parallelizing approximate pattern matching has not been done earlier, to the best of our knowledge.
Another natural approach for exploiting parallelism would be distributing the patterns to be matched onto the different processors and answer them in parallel; this is more of a load balancing problem and cannot be compared with our approach.
Parallel construction of text indices is another road of research \cite{DBLP:conf/icalp/FarachM96,DBLP:conf/cpm/KarkkainenKP15a}, and could easily be combined with our approach.
Finally, in the early 1990's, some work has been done on parallelizing online pattern matching algorithms \cite{DBLP:journals/siamcomp/BreslauerG90,DBLP:journals/siamcomp/BreslauerG92}.

\section{Preliminaries}
\label{sec:preliminaries}

Let $T=t_1\dots t_{n}$ be a \emph{text} of length $n$ consisting of characters contained in an integer \emph{alphabet} $\Sigma$ of size $\sigma=\abs{\Sigma} = n^{\Oh{1}}$. 
$\SUBSTR{i}{j}$ represents the \emph{substring} $t_i\dots t_j$ for $1\leq i\leq j\leq n$. 
We call $\SUBSTR{i}{n}$ the $i$-th \emph{suffix} of $T$ and $\SUBSTR{1}{i}$ the $i$-th \emph{prefix} of $T$. 
We denote the length of the \emph{longest common prefix} of the $i$-th and $j$-th suffix, i.e., $\SUBSTR{i}{n}$ and $\SUBSTR{j}{n}$, by \lcp[i,j].
The \emph{suffix array} of a text $T$ of length $n$ is a permutation of $\menge{1,\dots,n}$ such that 
$\SUBSTR{\SA{i}}{n}$ is lexicographically smaller than $\SUBSTR{\SA{i+1}}{n}$ for all $i=1,\dots,n-1$. 
We denote the inverse of \SA{} with $\SA{}^{-1}$.

An \emph{interval} $\Int{}\ = [b..e]$ is the set of consecutive integers from $b$ to $e$, for $b\le e$. 
For an interval $\Int{}$, we use the notations $\B{\Int{}}$ and $\E{\Int{}}$ to denote the beginning and end of $\Int{}$;
i.e., $\Int{}\ = [\B{\Int{}}..\E{\Int{}}]$.
We write $\abs{\Int{}}$ to denote the length of $\Int{}$; i.e., $\abs{\Int{}}=\E{\Int{}}-\B{\Int{}}+1$.

For a pattern $\alpha \in \Sigma^*$, let $\Int{\alpha}$ be the interval with $\SUBSTR{\SA{i}}{\SA{i}+\CARD{\alpha}-1}=\alpha\iff i\in \Int{\alpha}$.
If we consider two intervals $\Int{\alpha}$ and $\Int{\beta}$ and the corresponding merged interval $\Int{\alpha\beta}$, we call $\Int{\alpha}$ the \emph{left side} interval, $\Int{\beta}$ the \emph{right side} interval.
Let $\Psi^k\lbrack i\rbrack=\SA{}^{-1}\I{\SA{i}+k}$ be the position of the suffix~$\SUBSTR{\SA{i}+k}{n}$ in the suffix array.

\subsection{Suffix Trees}
The \emph{suffix tree} of a text~$T$ is the tree obtained by compacting the trie of all suffixes of~$T$;
it has $n$ leaves and less than $n$~internal nodes, where $n$ is the length of~$T$.
Each edge is labeled with a string.
We enumerate the leaves from left to right such that the $i$-th leaf has $i-1$ lexicographically preceding suffixes;
we write $\leafrank[\ell] = i$ if the leaf~$\ell$ is the $i$-th leaf.
We extend the notion of intervals to nodes; i.e., $\Int{v}$ denotes the interval $[b..e]$ such that $\SA{b},\dots,\SA{e}$ are exactly the suffixes below node $v$.

Since we target small space bounds, our focus is on a compressed representation of suffix trees \cite{sadakane07compressed,ohlebusch10cst++,fischer11combined,fischer09faster}.
The main ingredient of the so-called compressed suffix tree is a \emph{compressed suffix array} \cite{navarro07compressed}.
Depending on its implementation, a compressed suffix array takes \spaceSA{} bits of space, 
and gives $\timeSA$ time access to $\SA{}$ and $\SA{}^{-1}$  -- see Figure~\ref{figSA} for a comparison of the uncompressed and a compressed suffix array.
With additional $\Oh{n}$ \cite{sadakane07compressed} or even $o(n)$ bits \cite{fischer10wee}, a compressed suffix tree can answer queries on the LCP-array that stores the values \lcp[{\SA{i}},{\SA{i+1}}] for each $1 \le i \le n-1$.
The last ingredient of a compressed suffix tree is the tree topology (either explicitly \cite{sadakane07compressed} or implicitly \cite{ohlebusch10cst++}), and $o(n)$-bit succinct data structures for navigating in it \cite{sadakane10fully,gog10advantages}.

For our purpose, we need the following queries on the suffix tree:
\lca[u,v] returns the lowest common ancestor of two nodes~$u$ and $v$,
\lab[e] returns the label of an edge~$e$,
\pathlabel[v] returns the labels on the edges of the path from the root to~$v$.
These queries can be answered by all common compressed suffix trees~\cite{ohlebusch10cst++,sadakane07compressed,fischer11combined,fischer09faster}.

\begin{figure}[t]
  \centerline{
    \begin{tabular}{cccc}
      \toprule
      $\spaceSA$                                                               & \timeSA{}    &  reference
      \\\midrule
      $2n \lg n$                                                               & {\Oh{1}}     &  \cite{manber93suffix}
      \\\midrule
      $(1+\epsilon) n \lg n$                                                   & {\Oh{\epsilon}}     &  \cite{Munro2012Sro}
       \\\midrule
      $n H_k + {\smallOh{n H_k}} + {\Oh{n + \sigma^{k+1}\lg n + n\lg n / s }}$ & {\Oh{\lg s}} &  \cite{BNtalg13}
      \\\bottomrule
    \end{tabular}
  }
  \caption{Different representations of the compressed suffix array using \spaceSA{}~bits with the time bound \timeSA{} for accessing a value of \SA{} and $\SA{}^{-1}$. The sampling rate $s$ satisfies $s = \omega(\lg_\sigma n)$.}
  \label{figSA}
\end{figure}

\subsection{Integer Dictionaries}
An \emph{integer dictionary} is a set consisting of tuples of the form $(k,v)$, where $k \in U := \I{1..\abs{U}}$ is an integer from a universe~$U$ with $\abs{U} = n^{\Oh{1}}$; we call $k$ a \emph{key} and $v$ a \emph{value}.
A common task is to find a tuple in a dictionary by a given key.
Besides, we might be interested in finding the \emph{successor} (\emph{predecessor}) of a key~$k$, i.e.,
the largest (smallest) key~$k'$ in the dictionary with $k' \le k$ ($k' \ge k$).
We define the operations $\Key{(k,v)} = k$ and $\Value{(k,v)} = v$.

A well-known \emph{dynamic} integer dictionary representation is the \emph{$y$-fast trie}~\cite{willard83loglogarithmic}.
It can perform lookups, predecessor and successor queries in $\Oh{\lg \lg n}$ expected time, 
and uses \Oh{n \lg n} bits of space for storing $n$~elements.
It consists of an $x$-fast trie whose leaves store binary search trees.
In more detail, the $x$-fast trie stores \Oh{n / \lg n} entries in \Oh{\lg n} hash tables, and each leaf stores \Oh{\lg n} entries in its balanced binary search tree.
Here, we only need a \emph{static} version.
Therefore, we use perfect hashing~\cite{fredman84storing} as our hashing method, 
resulting in \Oh{\lg \lg n} time w.h.p.\ in worst case for all queries, 
while keeping the same space bounds and linear deterministic construction time.
Alternatively, we can construct the hash tables in \Oh{n \lg \lg n} deterministic time~\cite[Theorem~1]{ruzic08constructing}, resulting in \Oh{\lg \lg n} deterministic worst case time for all queries.
Further, we exchange the balanced binary search trees with sorted arrays, which will be useful later when we parallelize the queries.

\section{Suffix Array Interval Merging}
\label{sec:suffix_interval_merging}
To perform the merging of two suffix array intervals in $\Oh{\timeSA \lg \lg n}$ time,
we adapt the idea from Lam et al.~\cite[Lemma 19]{Lam2007}.
In their method, the aim is to output all occurrences resulting from the merging of two suffix array intervals in $\Oh{\timeSA (\lg \lg n + \occ)}$ time.
Here, we show how to modify their approach such that only the resulting interval is returned, leading to $\Oh{\timeSA \lg \lg n}$ time. Although our method is similar to Lam et al.\ \cite{Lam2007}, we give the full proof for completeness.

The idea is to sample the $\Psi$- and $\lcp$-values of each $(\lg^2 n)$-th suffix array position.
The sampling is stored in $y$-fast tries such that a search in a sorted array can be broken down to a $y$-fast trie query, 
or to a binary search on a range of size \Oh{\lg^2 n} -- both can be performed in \Oh{\lg \lg n} time.
To lower the space consumption, the sampling is done only for certain nodes determined by the heavy path decomposition of the suffix tree,
whose definition follows.

\subsection{Heavy Path Decomposition}
The \emph{heavy path decomposition} of a rooted tree assigns a \emph{level} to each node of the tree.
The level of the root is~$1$.
A node inherits the level of its parent if its subtree is the largest among the subtrees of all its siblings (ties are broken arbitrarily); we call such a node \emph{heavy}.
Otherwise, it has the level of its parent incremented by one; we then call the node \emph{light}.
Further, we define the root to be light.
A maximal connected subgraph consisting of nodes on the same level is called \emph{heavy path}.
A heavy path starts with a light node, called \emph{head},
and ends at a heavy leaf.

\subsection{Precomputed Data Structures}
We first present a simple data structure for the child-operation \child[u,c] in a (compressed) suffix tree, i.e., for finding the child $v$ of $u$ such that the label of the edge between $u$ and $v$ starts with character $c$. We use $\Delta=\Omega(\lg n)$ as the sampling rate throughout this section.

\begin{lemma}\label{lemChildChar}
The suffix tree of a text of length~$n$ can be augmented with a data structure of size $\Oh{n \lg n/\Delta}$ bits answering \child[v,c] in \Oh{\timeSA \lg \Delta} time. 
\end{lemma}
\begin{proof}
  We sample the children of each internal node~$u$ and store the sampled children in a $y$-fast trie with the first character of the edge label 
  between $u$ and the respective child as key.
  Given a node~$u$ with $k$~children, we sample every $\Delta$-th child of~$u$ so that
  $u$'s $y$-fast trie contains $\gauss{k/\Delta}$ elements.
  Since the suffix tree has less than $2n$ nodes, storing the $y$-fast tries for all internal node takes $\Oh{n \lg n /\Delta}$ bits overall.

  Given a character~$c$, we search \child[u,c] in the following way:
  Since the children of a node~$u$ are sorted by the first character of the edge connecting~$u$ with its respective child, 
  the $y$-fast trie of~$u$ can retrieve the first child~$v$ whose edge label~$\lab[u,v]$ is lexicographically at least as large as~$c$.
  If $c$ is a prefix of \lab[u,v], then we are done.
  Otherwise, say that $v$ is the $i$-th child of $u$, we can find \child[u,c] by a binary search on the range between the $(i-\Delta)$-th child and 
  the $i$-th child in $\Oh{\timeSA \lg \Delta}$ time.
\end{proof}

\newcommand{\Klab}{k}
\newcommand{\IL}{\ensuremath{i_{\textup{l}}}}
\newcommand{\IR}{\ensuremath{i_{\textup{r}}}}
\newcommand{\JL}{\ensuremath{j_{\textup{l}}}}
\newcommand{\JR}{\ensuremath{j_{\textup{r}}}}
\newcommand{\KL}{\ensuremath{k_{\textup{l}}}}
\newcommand{\KR}{\ensuremath{k_{\textup{r}}}}
\newcommand{\plcp}[2]{\lcp[{\SA{#1}},{\SA{#2}}]}

We also need a simple \Oh{n}-bit data structure to find the heavy leaf of a given heavy path in constant time~\cite[Lemma 15]{Lam2007}.

Next, we define three types of integer dictionaries that we are going to index in $y$-fast tries to allow fast lookups. 
For every light node~$v$, we define the integer dictionary
\[
\Head{v} := \menge{(\Psi^{\abs{\pathlabel[v]}}[i],i) : i \equiv 1 \!\!\!\pmod \Delta \text{~and~} i \in \Int{v}}.
\]
Given a heavy leaf~$\ell$ and its head~$v$, we define the two integer dictionaries
\[
  \LLeaf{\ell} := \menge{(\plcp{\leafrank[\ell]}{i},i) : i \equiv 1 \!\!\!\pmod \Delta \text{~and~} i \in \Int{v} \text{~and~} i \le \leafrank[\ell] }
\]
  and
  \[
    \RLeaf{\ell} := \menge{(\plcp{\leafrank[\ell]}{i},i) : i \equiv 1 \!\!\!\pmod \Delta \text{~and~} i \in \Int{v} \text{~and~} i > \leafrank[\ell] }.
\]
We store \Head{v} in a $y$-fast trie for each light node~$v$,
\LLeaf{\ell} and \RLeaf{\ell} in a $y$-fast trie for each heavy leaf~$\ell$.
Given an interval~$\IntJ{} \subseteq [1..n]$, 
we can find 
\begin{itemize}
  \item an $i \in \Head{v}$    with $\B{\IntJ} \le \Key{i}   = \Psi^{\abs{\pathlabel[v]}}[\Value{i}]        \le \E{\IntJ}$,
  \item an $\IL \in \LLeaf{\ell}$ with $\B{\IntJ} \le \Key{\IL} = \plcp{\leafrank[\ell]}{\Value{\IL}} \le \E{\IntJ}$, and
  \item an $\IR \in \RLeaf{\ell}$ with $\B{\IntJ} \le \Key{\IR} = \plcp{\leafrank[\ell]}{\Value{\IR}} \le \E{\IntJ}$,
\end{itemize}
all in $\Oh{\timeSA \cdot \lg \Delta}$ time.

\begin{lemma}
  We need $\Oh{n\lg^2 n/ \Delta}$ bits of space to store the $y$-fast tries for all $\Head{\cdot}$, $\LLeaf{\cdot}$, and $\RLeaf{\cdot}$.
\end{lemma}
\begin{proof}
Since the subtrees of the light nodes on the same level are disjoint,
summing over the sizes of \Head{v} for all light nodes~$v$ on the same level yields at most $n/\Delta$ elements.
Since the heavy path decomposition has at most $\Oh{\lg n}$ different levels
and a $y$-fast trie uses $\Oh{\lg n}$ bits per stored element, 
the claim for $\abs{\Head{v}}$ follows.

\newcommand{\ellH}{\ell_{\textup{H}}}
We analyze the size of \LLeaf{\cdot} by identifying a leaf with its \leafrank{}.
The sampling of~$\LLeaf{\cdot}$ considers only $n/\Delta$ leaves.
A leaf~$\ell$ has at most $\Oh{\lg n}$ light nodes as ancestors. 
So there are at most $\Oh{\lg n}$ heavy leaves~$\ellH$ having \leafrank[\ell] as a value in their dictionary~$\LLeaf{\ellH}$.
Hence, summing over $\LLeaf{\ellH}$ for all heavy leaves~$\ellH$ yields $\Oh{n \lg n/ \Delta}$ elements.
The same considerations lead to the same size bounds for~\RLeaf{\cdot}.
\end{proof}

\begin{lemma}\label{lem:suffix_interval_merging_time_delta}
  Given the compressed suffix tree of~$T$ and the dictionaries $\Head{\cdot}$, $\LLeaf{\cdot}$ and $\RLeaf{\cdot}$ as defined above, we can merge two suffix array intervals in \Oh{\timeSA \lg \Delta}~time.
\end{lemma}

\begin{proof}
  Let $\Int{\alpha}$ and $\Int{\beta}$ be two suffix array intervals and $P := \alpha \beta$.
  Our task is to search the interval $\Int{P} \subseteq \Int{\alpha}$ with 
  $\Psi^{\abs{\alpha}}[i] \in \Int{\beta}$ for all $i \in \Int{P}$.
  Since $i \mapsto \Psi^{\abs{\alpha}}[i]$ is monotonically increasing for $i \in \Int{\alpha}$, 
  the merge could be solved with two binary searches in~$\Int{\alpha}$.
  To obtain the $\Oh{\timeSA \lg \Delta}$ time bound we will either use the $y$-fast tries, or perform a binary search on $\Oh{\Delta}$-large intervals.

  Let us take the node $v$ whose suffix array interval is \Int{\alpha}, i.e., 
  the lowest common ancestor of the leaves with \leafrank[]~$\B{\Int{\alpha}}$ and~$\E{\Int{\alpha}}$.
  We consider two cases:

  \begin{description}
  \item[Node $v$ is heavy.]
  Let $H$ be the heavy path to which $v$ belongs, $\ell$ its heavy leaf, and $u$ its head.

  If $\Head{u}$
  is empty, there are less than $\Delta$ leaves in the subtree rooted at $u$.
  Since $\Int{P} \subset \Int{u}$, we can find $\Int{P}$ by binary search in \Oh{\timeSA \lg \Delta}.

  Otherwise ($\Head{u} \not= \emptyset$), 
  let $q := \plcp{\Psi^{\abs{\alpha}}[{\leafrank[\ell]} ]}{\B{\Int{\beta}}}$.
  The value~$q$ is the length of the longest common prefix of~$P$ 
  and the path label of~$\ell$, subtracted by $\abs{\alpha}$.
  By definition of~$q$, there is a node~$r$ on~$H$ whose path label coincides with $\alpha\beta[1..q]$.
  In particular, this is the node on the path~$H$ whose path label is the longest prefix of~$P$ with respect to the path labels of all other nodes on $H$.
  Since $\Int{P} \subset \Int{r}$, our task is to find~$r$ in $\Oh{\timeSA \lg \Delta}$ time.
  To this end, we locate a leaf whose LCA with~$\ell$ is~$r$.

  The interval boundaries can be found by a coarse search on the $y$-fast tries of~$\LLeaf{\ell}$ and~$\RLeaf{\ell}$,
  and a subsequent refinement step using binary search.
Let $\Klab := \leafrank[\ell]$.
Since $i \mapsto \plcp{i}{k}$ is monotonically increasing for $i < \Klab$, and monotonically decreasing for $i > \Klab$, 
we can perform the binary search for a value on the key-sorted integer dictionaries \menge{(\plcp{i}{\Klab},i) : i < \Klab} and 
\menge{(\plcp{i}{\Klab},i) : i > \Klab} conceptionally.
The $y$-fast tries at \LLeaf{\ell} and \RLeaf{\ell} help us computing
  the tuple
  $\JL \in \LLeaf{\ell} \cup\menge{(\abs{\pathlabel[\ell]},\Klab)}$ with 
  \[
    \plcp{\Value{\JL} - \Delta}{\Klab} \le \abs{\alpha} + q \le \Key{\JL} = \plcp{\Value{\JL}}{\Klab}
  \]
  and the tuple
  $\JR \in \RLeaf{\ell} \cup\menge{(\abs{\pathlabel[\ell]},\Klab)}$ with 
  \[
    \Key{\JR} = \plcp{\Value{\JR}}{\Klab} \le \abs{\alpha} + q \le \plcp{\Value{\JR} + \Delta}{\Klab}.
  \]
  Since $\plcp{\Value{\JL} - \Delta}{\Klab} \le \abs{\alpha} + q \le \plcp{\Value{\JR} + \Delta}{\Klab}$, 
  we can find one of the positions $\IL \in [\Value{\JL} - \Delta .. \Value{\JL}]$ and $\IR \in [\Value{\JR} .. \Value{\JR} + \Delta]$ by binary search
  such that $\plcp{\IL}{\Klab} = \plcp{\IR}{\Klab} = \abs{\alpha} + q$.
  The binary search takes $\Oh{\timeSA \lg \Delta}$ time.
  On finding $\IL$ or $\IR$, we can retrieve $r$, i.e., the lowest common ancestor of $\ell$ and the $\IL$-th or $\IR$-th leaf.
  If the pattern~$P$ is a prefix of the path label of~$r$, then $\Int{P} = \Int{r}$, and we are done.
  Otherwise, we choose the child~$w$ of~$r$ whose edge label~$S$ starts with $\beta[q+1]$;
  $w$ can be retrieved in $\Oh{\timeSA \lg \Delta}$ time by Lemma~\ref{lemChildChar}.
  The child~$w$ must be a light node, for otherwise we get a contradiction to the definition of~$r$.
  We set $v \gets w$, $\alpha \gets P[1..\abs{\alpha}+q+\abs{S}]$, $\beta \gets P[\abs{\alpha}+1..\abs{P}]$, and jump to the next case:

  \item[Node $v$ is light.]
  If $\Head{v}$ is empty, then $\abs{\Int{v}} < \Delta$.
  Therefore, we can find the interval boundaries of $\Int{P}$ in \Int{v} with a binary search in \Oh{\timeSA \lg \Delta} time.
  Otherwise, 
  we use the $y$-fast trie storing~$\Head{v}$ to find
  the tuple $\JL \in \Head{v}$ with the smallest key satisfying $\B{\Int{\beta}} \le \Key{\JL} = \Psi^{\abs{\alpha}}[\Value{\JL}]$ and the tuple
  $\JR \in \Head{v}$ with the largest key satisfying $\Key{\JR} = \Psi^{\abs{\alpha}}[\Value{\JR}] \le \E{\Int{\beta}}$.
  If both exist, we can find the positions $\B{\Int{P}} \in [\Value{\JL}-\Delta .. \Value{\JL}]$ and $\E{\Int{P}} \in [\Value{\JR} .. \Value{\JR}+\Delta]$ by two binary searches.
  If there is no tuple~$i \in \Head{v}$ with $\B{\Int{\beta}} \le \Key{i} \le \E{\Int{\beta}}$, we search with the $y$-fast trie of \Head{v}
  the tuple~$\KL \in \Head{v}$ with 
  \[
  \Key{\KL} = \Psi^{\abs{\alpha}}[\Value{\KL}] \le \B{\Int{\beta}} \le \Psi^{\abs{\alpha}}[\Value{\KL} + \Delta]
  \]
  and the  tuple~$\KR \in \Head{v}$ with 
  \[
  \Psi^{\abs{\alpha}}[\Value{\KR} - \Delta] \le \E{\Int{\beta}} \le \Psi^{\abs{\alpha}}[\Value{\KR}] = \Key{\KR}.
  \]
  Both values exist, and $\Value{\KR} - \Value{\KL} \le \Delta$. 
  So we find the interval~\Int{P} by applying two binary searches to the range $\Value{\KL} .. \Value{\KR}$.
\end{description}  
\end{proof}

Setting $\Delta := \lg^{c} n$ for $c \ge 2$ yields:
\begin{lemma}
  \label{lem:suffix_interval_merging_time}
  Given the compressed suffix tree of~$T$, there is a data structure of size \Oh{n}~bits that allows us to merge two suffix array intervals in \Oh{\timeSA \lg \lg n}~time.
\end{lemma}

\section{Parallel Exact Pattern Matching}
\label{sec:parallelization}
We parallelize the merging of suffix array intervals that we presented in Section~\ref{sec:suffix_interval_merging} and show that queries in the suffix tree using consecutive subpatterns and linear space can be solved in parallel on a CREW-PRAM\@. 
For this, we use parallel binary search:

\begin{lemma}[{\cite[Theorem 2.1]{Snir1985}}]\label{lemParallelBinSearch}
  Given a sorted array of size~$n$, a binary search requires \Oh{1+\lg_p n} time when operating on a CREW-PRAM with $p$~processors.
\end{lemma}

We conclude that we can parallelize the query on $y$-fast tries in the same way:

\begin{lemma}
  \label{lem:parallel_y_fast_trie}
  A $y$-fast trie
  can do lookups, predecessor and successor queries 
  in $\Oh{1+\lg_p\lg n}$ time using $p$ processors.
\end{lemma}
\begin{proof}
  We can find an element in an $x$-fast trie in $\Oh{1+\lg_p\lg n}$ time using parallel binary search (Lemma~\ref{lemParallelBinSearch}) on the $\Oh{\lg n}$ hash tables.
  The sorted arrays stored at the leaves can similarly be searched in \Oh{1+\lg_p\lg n} time, again using Lemma~\ref{lemParallelBinSearch}.
\end{proof}

Let us focus on the merging of two suffix array intervals as treated in Section~\ref{sec:suffix_interval_merging}.
The dominant term of its running time is due to the query time of the $y$-fast tries and the binary searches.
As we can parallelize both, a parallelization of the merging algorithm improves the time bounds significantly:
\begin{lemma}
  \label{lem:parallel_merging}
  Given $p$ processors and two suffix array intervals \Int{\alpha} and \Int{\beta}, the merged interval \Int{\alpha\beta} can be computed 
  in $\Oh{\timeSA (1+\lg_p\lg n)}$ time 
  and\linebreak $\Oh{\timeSA\min\left(p,\lg n\right) (1+\lg_p\lg n)}$ work.
\end{lemma}
\begin{proof}
  We can merge two suffix array intervals in $\Oh{\timeSA \lg\lg n}$ time using Lemma \ref{lem:suffix_interval_merging_time}. 
  Recalling the proof of Lemma~\ref{lem:suffix_interval_merging_time_delta}, we took the node $v$ whose suffix array interval is $\Int{\alpha}$. 
  There, in both cases ($v$ is either heavy or light), the time is dominated by searching in $y$-fast tries, and/or by binary searching in \Oh{\Delta} sampled $\Psi$- or \lcp{}-values.
  Both can be parallelized by Lemmas \ref{lemParallelBinSearch} and \ref{lem:parallel_y_fast_trie}, respectively. This yields 
  $\Oh{\timeSA (1+\lg_p\lg n)}$ time using $p$ processors. 
  During the parallel searches, we use at most \Oh{\lg n} processors \Oh{1+\lg_p\lg n} times. 
  This amounts to $\Oh{\timeSA\min\left(p,\lg n\right) (1+\lg_p\lg n)}$ work.
\end{proof}

Being able to merge two suffix array intervals in parallel, we now show how to compute the suffix array interval of a pattern $P$ in parallel.
To this end, we decompose the pattern in subpatterns $\alpha_1,\dots,\alpha_p$ such that $P=\alpha_1\alpha_2\dots\alpha_p$, and then compute the suffix array intervals for the subpatterns.
Then we merge those intervals in parallel.
\begin{theorem}
  Given a text of size $n$ and a pattern of size $m$. With $p \le m$ processors, we can compute the suffix array interval of the pattern in 
  $\Oh{\timeSA \left(\frac{m}{p}+ \lg p + \lg\lg p\cdot\lg\lg n\right)}$ time and $\Oh{\timeSA\left( m+ \min\left(p,\lg n\right)\left(\lg p + \lg\lg p\cdot\lg\lg n\right)\right)}$ work.
    In order to achieve this time bound we need an index of size $\spaceSA + \Oh{n}$ bits.
  \label{the:parallel_exact}
\end{theorem}
\begin{proof}
  Let $P=\alpha^0_1\alpha^0_2\dots\alpha^0_p$ be a pattern of length $m$ such that $\CARD{\alpha^0_i}=\Theta{\left(\frac{m}{p}\right)}$ for $i=1,\dots,p$.
  The computation of all intervals $\Int{\alpha^0_i}$ requires \Oh{\timeSA \frac{m}{p}} time.
  In the first merge step we have two processors to compute each of the intervals $\Int{\alpha^1_i}:=\Int{\alpha^0_{2i-1}\alpha^0_{2i}}$ for $i=1,\dots,\frac{p}{2}$.
In each merge step we halve the number of intervals. So in the $k$-th merge step ($1 \le k \le \lg p$), we have $2^k$ processors
to compute each of the intervals $\Int{\alpha^k_i}:=\Int{\alpha^{k-1}_{2i-1}\alpha^{k-1}_{2i}}$ for $i=1,\dots,\frac{p}{2^k}$.
As we require $\Oh{\lg p}$ merge steps and can use Lemma~\ref{lem:parallel_merging} with $2^k$ processors in the $k$-th merge step, 
the interval $\Int{P}$ can be computed in
$\Oh{\timeSA \sum_{k=1}^{\lg p} \left(1+\lg_{2^k}\lg n\right)} =
\Oh{\timeSA \left( \lg p + \lg\lg p \cdot \lg\lg n\right)}$ time, 
given the intervals $\Int{\alpha^0_i}$ of the subpatterns -- see Figure~\ref{fig:exact_merging}.
In total, $\Int{P}$ can be found in
$\Oh{\timeSA \left(\frac{m}{p} + \lg p + \lg\lg p\cdot\lg\lg n\right)}$ time.

During the computation of the suffix array intervals of the subpatterns of $P$ we use all $p$ processors, which results in $\Oh{\timeSA m}$ work. The same holds for each merging step, as we use all processors to parallelize the binary search. We have $\Oh{\lg p}$ merge steps. During the $k$-th merge step, we merge $\frac{p}{2^k}$ suffix array intervals with $2^k$ processors each. Using Lemma~\ref{lem:parallel_merging} the total work 
is
$%
\Oh{\timeSA\left( m+ \min\left(p,\lg n\right)\left(\lg p + \lg\lg p\cdot\lg\lg n\right)\right)}$.
\end{proof}

\begin{figure}[t!]
  \centering
  \includegraphics[scale=.8]{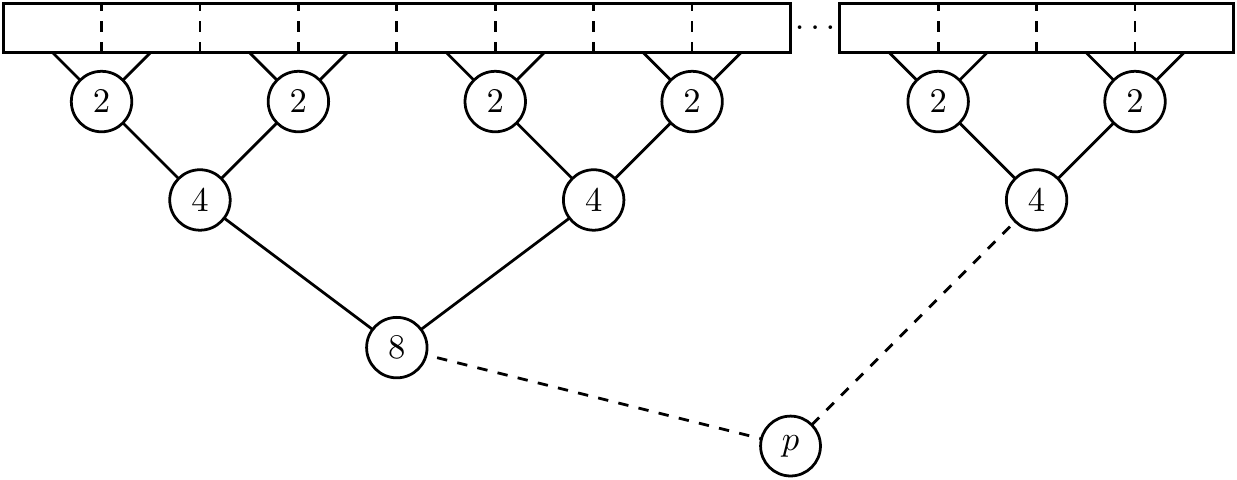}
  \caption{Schedule for the merging of $p$ suffix array intervals, i.e., the suffix array intervals of the subpatterns. The number in each node denotes the number of processors available for the merging of the two considered suffix array intervals.\label{fig:exact_merging}}
\end{figure}

\section{Parallel Approximate Pattern Matching}
\label{sec:parallel_approximate_pattern_matching}
In this section, we consider two different distances for the approximate string matching problem. The first distance we consider is the \emph{Levenshtein distance}, where the distance between two patterns $P$ and $P^\prime$ is the minimal number of the operations \emph{insert}, \emph{change} and \emph{remove} required to change $P^\prime$ into $P$. The second one is the \emph{Hamming distance}, where the distance of two pattern $P$ and $P^\prime$ of equal length is the number of mismatching positions, i.e., $\CARD{\lbrace i\colon P\lbrack i\rbrack\neq P^\prime\lbrack i\rbrack, 1\leq i\leq\abs{P}\rbrace}$. We consider two problems related to these distances.

\begin{description}
  \item[$k$-difference problem] Given a text $T$ of length $n$ and a pattern $P$ of length $m$, we want to report all occurrences $i\in\menge{1,\dots,n}$ such that $T\lbrack i..i+j\rbrack$ and $P$ have a Levenshtein distance of at most $k$ for at least one $j\in\menge{0,\dots,n-i}$. 
  \item[$k$-mismatch problem] Given a text $T$ of length $n$ and a pattern $P$ of length $m$, we want to report all occurrences $i\in\menge{1,\dots,n-m}$ such that $T\lbrack i..i+m\rbrack$ and $P$ have a Hamming distance of at most $k$.
\end{description}

We apply the results from Section~\ref{sec:suffix_interval_merging} to parallelize the approximate string matching algorithm by Huynh et al. \cite{Huynh2006}. To do so, we first present an approach to compute the suffix array intervals of all prefixes and suffixes of the pattern in parallel -- see Figure~\ref{fig:approximative_preprocessing}.

\begin{lemma}\label{lemParallelPrecompute1Approx}
  Given a text of length~$n$ and its suffix array, 
  we can compute the suffix array intervals of all prefixes and suffixes of a pattern of length~$m$ in 
  $\Oh{\timeSA (1+\frac{m}{p})\lg p\cdot\lg\lg n}$ parallel time operating on a CREW-PRAM with $p$~processors.
\end{lemma}
\begin{proof}
  Let $P=\alpha_0\alpha_1\dots\alpha_{p-1}$ be a pattern of length $m$ such that $\CARD{\alpha_i}=\frac{m}{p}$ for $i=0,\dots,p-1$. Thus, the $j$-th prefix of a subpattern $\alpha_i$ is $P[1+i\frac{m}{p}..i\frac{m}{p}+j]$ for all $i=0,\dots,p-1$ and $j=1,\dots,\frac{m}{p}$. First, we compute the suffix array intervals for all those prefixes
  in parallel, which requires $\Oh{1+m/p}$ time, as no merging is necessary during this step.

  In the second step, we merge the suffix array intervals in parallel. Since we want the suffix array intervals of all prefixes of the pattern, during the first merge step we merge the suffix array intervals $\Int{P\lbrack1+2i\frac{m}{p}..(2i+1)\frac{m}{p}\rbrack}$ as the left side interval with each of the intervals $\Int{P\lbrack1+(2i+1)\frac{m}{p}..(2i+1)\frac{m}{p}+j}$ as right side interval for all $i=0,\dots,\frac{p}{2}-1$ and $j=1,\dots,\frac{m}{p}$. This results in the intervals $\Int{P\lbrack1+2i\frac{m}{p}..(2i+1)\frac{m}{p}+j}$ for $i=0,\dots,\frac{p}{2}-1$ and $j=1,\dots,\frac{m}{p}$. During each merge step, we halve the number of left side intervals that we have to consider during the next merge step but double the number of right side intervals that are merged, i.e., in the $k$-th merge step, we merge the intervals $\Int{P\lbrack1+2^k i\frac{m}{p}..(2^k i+2^{k-1})\frac{m}{p}\rbrack}$ with each of the intervals $\Int{P\lbrack1+(2^k i+2^{k-1})\frac{m}{p}..(2^k i+2^{k-1})\frac{m}{p}+j\rbrack}$ for $i=0,\dots,\frac{p}{2^k}-1$ and $j=1,\dots,2^{k-1}\frac{m}{p}$. 
  This amounts to $\Oh{m}$ intervals that need to be merged in each step. 
  In the end, we obtain the suffix array intervals of the prefixes of $P$, i.e., the intervals $\Int{P\I{1..j}}$ for $j=1,\dots,m$.
  Since we start with $p$ left side intervals, and each merge step halves the number of left side intervals, we end up with $\lg p$ merge steps.

The computation of the suffix array intervals of all suffixes of $P$ works analogously. Using Lemma \ref{lem:suffix_interval_merging_time}, we can compute the suffix array intervals of all prefixes and suffixes of the pattern in $\Oh{\timeSA (1+\frac{m}{p})\lg p\cdot\lg\lg n}$ time.
\end{proof}
\begin{figure}[t!]
  \centering
  \includegraphics{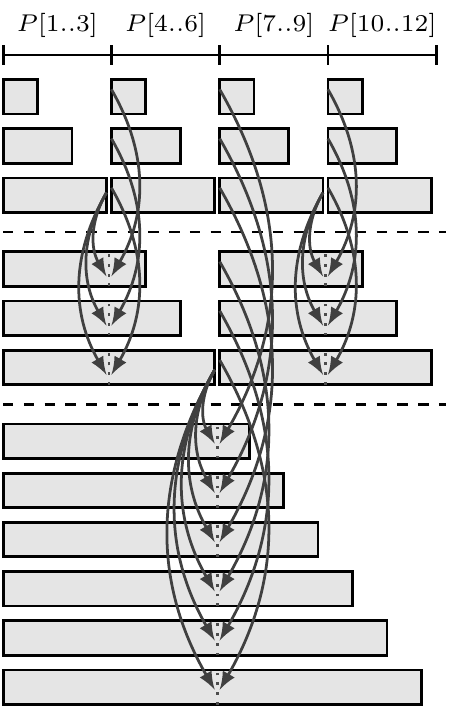}
  \caption{Schematics of the merging process to compute the suffix array intervals of all prefixes of a pattern $P=a_1\dots a_4$ of length $m=12$ using $p=4$ processors. The gray blocks above the first dashed line represent the suffix array intervals of all prefixes of the subpatterns $a_i$ (for $i=1,\dots,4$). The blocks between the dashed lines represent the suffix array intervals after the first merge step. The intervals that are merged are shown by arrows. The blocks below the second dashed line are the suffix array intervals computed in the second merge step.\label{fig:approximative_preprocessing}}
\end{figure}

\begin{figure}[t]
  \centerline{
    \begin{tabular}{llll}
      operation    & $c$                                   & $P'$                  & intervals to merge
      \\\toprule
      substitution & $c \in \Sigma \setminus \menge{P[i]}$ & $P[1..i-1]cP[i+1..m]$ & {\Int{{\child[v,c]}}} and {\Int{P\I{i+1..m}}}
      \\\midrule
      deletion     & $-$                                   & $P[1..i-1]P[i+1..m]$  & {\Int{v}} and {\Int{P\I{i+1..m}}}
      \\\midrule
      insertion    & $c \in \Sigma$                        & $P[1..i-1]cP[i..m]$   & {\Int{{\child[v,c]}}} and {\Int{P\I{i..m}}}
      \\\bottomrule
    \end{tabular}
  }
  \caption{Let $P'$ be the resulting string of introducing an error in the pattern~$P[1..m]$ at position~$i$. 
    Further, let $v$ be the suffix tree node with $\Int{v} = {\Int{P\I{1..i-1}}}$. We can compute the two suffix array intervals considered for merging in \Oh{\timeSA \lg \lg n} time, and perform the merging in the same time.
  }
  \label{fig1Approx}
\end{figure}

\begin{theorem}
  \label{the:1_error_approximate_string_matching}
  With $\spaceSA + \Oh{n}$ bits of space, 
  the $1$-difference and $1$-mismatch problems can be solved in parallel in 
  $\Oh{\timeSA \lg \lg n \left( \frac{m\sigma}{p} + (1+\frac{m}{p}) \lg p\right) + \occ}$ for $p \le m\sigma$.
\end{theorem}
\begin{proof}
  We precompute the suffix array intervals $\Int{P\I{i..m}}$ and $\Int{P\I{1..i}}$ for all $1 \le i \le m$ in parallel by Lemma~\ref{lemParallelPrecompute1Approx}.
  This requires \Oh{\timeSA(1+\frac{m}{p}) \lg p\cdot\lg\lg n} time.
  The exact matches are found in the interval~$\Int{P\I{1..m}}$.
  To compute the matches with one error, we iterate over all positions in~$P[1..m]$, and introduce an error at one position $i$ with $1\leq i\leq m$.
  An error can be introduced by an insertion, a deletion, or a substitution.
  Let us fix one modification occurring at position~$i$, and call the modified string~$P'$.
  Our task is to find $\Int{P'}$.
  To this end, we exploit some already computed results, i.e., we have
  $\Int{P'\I{1..i-1}} = \Int{P\I{1..i-1}}$, and either
  (substitution) $\Int{P'\I{i+1..m}} = \Int{P\I{i+1..m}}$, 
  (deletion) $\Int{P'\I{i..m-1}} = \Int{P\I{i+1..m}}$, 
  or (insertion) $\Int{P'\I{i+1..m+1}} = \Int{P\I{i..m}}$ -- see Figure~\ref{fig1Approx}.
    If $P'$ resulted from an insertion or substitution, the interval~\Int{P'\I{1..i-1}} can be enhanced to $\Int{P'\I{1..i}}$ by \child[v,P'\I{i}] in \Oh{\timeSA \lg\lg n} time due to Lemma~\ref{lemChildChar}, where $v$ is the node with $\Int{v} = \Int{P'\I{1..i-1}}$.
    Finally, we can compute \Int{P'} by merging two intervals in \Oh{\timeSA \lg \lg n} time with Lemma \ref{lem:suffix_interval_merging_time}.
Introducing an error in $P$ at $m$ different positions with $\sigma$ different characters is embarrassingly parallel.
With $p\leq m\sigma$~processors it requires \Oh{\timeSA \frac{m \sigma}{p} \lg \lg n + \occ} time in addition to the time for the preprocessing.
\end{proof}

Up to now, we have assumed that the time for the output is in $\Oh{\occ}$. Unfortunately, this is not always the case, as an occurrence of a pattern with $k$ errors may be reported multiple times. For example, if we allow one error, the pattern \texttt{aba} could be reported twice at the first position of the text \texttt{aaa}, as the second position of the pattern could either be deleted or replaced. Hence, we need to make sure that each occurrence of a pattern is reported just once, regardless how many different combinations of operations can be used to change the pattern to the corresponding substring. This problem has been discussed and solved in~\cite{Huynh2006}.
\begin{lemma}[{\cite[Discussion related to Theorem 2]{Huynh2006}}]
  \label{lem:reporting_occurrences_only_once}
  Given a pattern $P$, we can check whether an occurrence of the pattern $P^\prime$ with at most $k$ errors is minimal regarding its distance and its edit operations to $P$ in $\Oh{k}$ time whenever we append a character or want to report an occurrence.
\end{lemma}

Using Lemmas~\ref{lemParallelPrecompute1Approx}~and~\ref{lem:reporting_occurrences_only_once}, we can solve the $1$-difference and $1$-mismatch problems in parallel as described above. The same is true for the $k$-difference and $k$-mismatch problems.

\begin{theorem}
  \label{the:k_error_approximate_string_matching}
  Using $\spaceSA + \Oh{n}$ bits of space, the $k$-difference and $k$-mismatch problems can be solved in parallel in 
  $\Oh{\timeSA \left(\frac{m^k \sigma^k}{p} \max\left(k,\lg\lg n\right)+ (1+\frac{m}{p})\lg p \cdot \lg \lg n\right) +\occ}$ for $p \le m^k \sigma^k$ processors.
\end{theorem}
\begin{proof}
  The idea of the algorithm is similar to the algorithm of Theorem~\ref{the:1_error_approximate_string_matching}. First, we compute the suffix array intervals of all the suffixes and prefixes of the pattern using Lemma~\ref{lemParallelPrecompute1Approx}.
  This requires \Oh{\timeSA(1+\frac{m}{p}) \lg p\cdot\lg\lg n} time.
  We want to introduce at most $k$ errors in parallel. Again, we parallelize over the positions of the introduced errors. Similar to the idea of Theorem~\ref{the:1_error_approximate_string_matching}, we merge different suffix array intervals. But in this case, we cannot parallelize over one position, instead we have to parallelize considering up to $k$ positions where we can include an error.

  The number of patterns $P^\prime$ that have a distance of at most $k$ from $P$ is bounded by $\Oh{\sigma^k m^k}$ \cite[Theorem 6]{Ukkonen1993}. Thus, we require $\Oh{\timeSA \frac{m^k \sigma^k}{p}\max\left(k,\lg\lg n\right)+\occ}$ time using $p \le \sigma^k m^k$ processors in parallel. The $\Oh{\max\left(k,\lg\lg n\right)}$-term results from the check of whether the occurrence is computed with minimal distance to the pattern $P$, which has to be done every time we update the considered pattern and requires $\Oh{k}$ time using Lemma~\ref{lem:reporting_occurrences_only_once}.
\end{proof}
\subsubsection*{Acknowledgment} We thank Anders Roy Christiansen, who found an error in a previous version of this paper.

\bibliographystyle{plain}
\bibliography{lit}

\end{document}